\documentclass[11pt, leqno, letterpaper]{amsart}
\usepackage{graphicx}    % needed for including graphics e.g. EPS, PS
\usepackage{bbm, xcolor}

\usepackage[margin=1.2in]{geometry}

\usepackage{amsmath,amsthm,amsfonts,amssymb}

\newtheorem{theorem}{Theorem}

\newtheorem{proposition}[theorem]{Proposition}

\begin{document}
 % Start your text

\title{Subleading correction to the Asian options volatility in the Black-Scholes model}
\author{Dan Pirjol}

\address
{School of Business\newline
\indent Stevens Institute of Technology\newline 
\indent Hoboken, NJ 07030}

\email{
dpirjol@gmail.com}
\date{5 July 2024}

\keywords{Asymptotic expansions, Asian options, time integral of the geometric Brownian motion}

\begin{abstract}
The short maturity limit $T\to 0$ for the implied volatility of an Asian option in the
Black-Scholes model is determined by the large deviations property 
for the time-average of the geometric Brownian motion.
In this note we derive the subleading $O(T)$ correction to this implied volatility,
using an asymptotic expansion for the Hartman-Watson distribution.
The result is used to compute subleading corrections to Asian options prices in 
a small maturity expansion, sharpening the leading order result obtained using 
large deviations theory.  
We demonstrate good numerical agreement with precise benchmarks 
for Asian options pricing in the Black-Scholes model.
\end{abstract}

\maketitle

\section{Introduction}
\label{sec:1}

Asian options are derivatives with payoff linked to the 
time average of the asset price
\begin{equation}\label{ATdef}
A_T := \frac{1}{T} \int_0^T S_t dt\,.
\end{equation}

We are interested in pricing Asian options under the Black-Scholes model where the asset price follows a geometric Brownian motion
\begin{equation}\label{BS}
\frac{dS_t}{S_t} = (r-q) dt + \sigma dW_t \,,
\end{equation}
with initial condition $S_0>0$. 
Asian options pricing under the Black-Scholes model has been widely studied, using and a large variety of approaches, both numerical and analytical.  See \cite{Fusai2008} for a survey of methods.
Restricting to analytical approaches we mention here the Geman-Yor approach \cite{Geman1993,DufresneReview}, the Laguerre polynomial expansion method \cite{Dufresne2000}, the PDE expansion
method \cite{Rogers1995,FPP2013}, and the spectral method \cite{Linetsky2004}.
We mention also the large- and small-strike
asymptotics of Asian options in the Black-Scholes model obtained in 
\cite{Gulisashvili2010} and \cite{Zhu2015}.

% Introduce the LDP and the asymptotics of the implied vol

The short maturity asymptotics of Asian option prices has been studied using
probabilistic methods from Large Deviations theory \cite{MLP,SIFIN,PZAsianCEV},
assuming that $S_t$ follows a one-dimensional diffusion
\begin{equation}\label{gBM}
dS_t= \sigma(S_t) S_t dW_t +  (r-q) S_t dt \,,
\end{equation}
under suitable technical conditions on the volatility function $\sigma(\cdot )$.
These results include as a limiting case the Black-Scholes model
corresponding to $\sigma(x)=\sigma$. 
The short maturity asymptotics of Asian options under stochastic volatility models has been studied recently using Malliavin calculus methods in \cite{Alos2022}. 

The leading short maturity asymptotics of the out-of-the-money Asian
option prices is given in Theorem 2 in \cite{SIFIN}. We quote the result 
applied to the Black-Scholes model. For simplicity of notation we take $q=0$ in the following.

\begin{theorem}\label{prop:SIFIN}
Assume that the asset price follows the Black-Scholes model 
$S_t = S_0 e^{\sigma W_t + (r-\frac12\sigma^2)t}$.

(i) for out-of-the-money call Asian options we have
\begin{equation}
\lim_{T\to 0} T \log C(K,T) = - \frac{1}{\sigma^2} J_{BS}(K/S_0)\,, K > S_0\,.
\end{equation}

(ii) for out-of-the-money put Asian options we have
\begin{equation}
\lim_{T\to 0} T \log P(K,T) = - \frac{1}{\sigma^2} J_{BS}(K/S_0)\,, K < S_0\,.
\end{equation}
\end{theorem}

The rate function $J_{BS}(k) $ is given in closed form in
Proposition 12 in \cite{SIFIN}
\begin{equation}\label{JBSsol}
J_{BS}(k) = \left\{
\begin{array}{cc}
\frac12\beta^2 - \beta \tanh \frac{\beta}{2} & \,, k \geq 1 \\
2\xi (\tan \xi - \xi) & \,, 0 < k \leq 1 \\
\end{array}
\right.
\end{equation}
where $\beta$ is the solution of the equation $\frac{\sinh \beta}{\beta} = k$ 
and $\xi $ is
the solution in $[0,\frac{\pi}{2}]$ of the equation 
$\frac{\sin 2\xi}{2\xi} = k$.

The analyticity properties of the function $J_{BS}(z)$ in the complex $z$ plane were
studied in \cite{Nandori2021}, see Sec.~4.1. This function has no singularities along the real positive axis, and the closest singularity to $z=1$ is a pole at $z=0$. 
For practical computations it is convenient to use the Taylor expansion of the rate function in powers of $\log k$. The first few terms of this expansion are
\begin{equation}\label{JTaylor}
J_{BS}(k) = \frac32 \log^2 k - \frac{3}{10} \log^3 k + \frac{109}{1{\small,}400} \log^4 k +
 O(\log^5 k) \,.
\end{equation}
The radius of convergence of this series is determined by the position of the 
singularities in the complex plane of the function $z\mapsto J_{BS}(e^z)$ 
and is $|\log k| < 3.49295$, see Proposition 4.1(ii) in \cite{Nandori2021}. 
Outside of this region the exact result (\ref{JBSsol}) must be used. 

%Similar results are obtained for Asian options with discrete time averaging in \cite{PZAsian}.

An Asian option with maturity $T$ and strike $K$ can be priced as an European 
option on a Black-Scholes asset, with the same maturity and strike, 
with an implied volatility $\Sigma_{\rm LN}(K,T)$ chosen such that 
\begin{equation}\label{CAsianBS}
C(K,T) = C_{BS}(K,T;A_{\rm fwd},\Sigma_{\rm LN}(K,T))
\end{equation}
where $C(K,T)$ is the Asian option price, and the forward price 
$A_{\rm fwd}$ is given by
\begin{equation}\label{Afwd}
A_{\rm fwd} := \frac{1}{T} \int_0^T \mathbb{E}[S_t] dt = 
S_0 \frac{e^{(r-q)T}-1}{(r-q)T} \,.
\end{equation}
We call the volatility $\Sigma_{\rm LN}(K,T)$ the \textit{equivalent log-normal volatility} of the Asian option, following Sec.~4.3 of \cite{PZAsian}. The representation 
(\ref{CAsianBS}) is also useful for computing Asian option sensitivities \cite{PZGreeks}.

The short-maturity asymptotics for the Asian option prices of 
Proposition~\ref{prop:SIFIN} is equivalent with a short-maturity asymptotics for the equivalent log-normal volatility
\begin{equation}\label{Sigma0def}
\lim_{T\to 0} \Sigma^2_{\rm LN}(K,T) =
\sigma^2 \frac{\log^2 (K/S_0)}{2J_{\rm BS}(K/S_0)} =:  \Sigma_0^2(K/S_0) \,.
\end{equation}
Using the expansion (\ref{JTaylor}) we get 
\begin{equation}
\Sigma_0^2(k) = \sigma^2 \frac{\log^2 k}{2J_{\rm BS}(k)} =
\frac13 \sigma^2 \left( 1 + \frac15 \log k - \frac{1}{84} \log^2 k - \frac{17}{10,\!500} \log^3 k + O(\log^4 k) \right)
\end{equation}
where $k=K/S_0$ and $J_{BS}(k)$ is the rate function appearing in the statement of Proposition~\ref{prop:SIFIN}.
See also Proposition 18 in \cite{SIFIN}.
A similar asymptotic result was obtained in \cite{SIFIN} in the more general setting of the local volatility model, and in 
\cite{PZAsianJD} for a class of jump-diffusion models with local volatility.

While our study of the equivalent log-normal volatility $\Sigma_{LN}(K,T)$ is limited to a short maturity expansion, we note an exact prediction which can be extracted from the results of Ref.~\cite{Zhu2015}. Proposition 1 in this paper implies the extreme strikes asymptotics \\
$\lim_{K\to \{0,\infty\}} \Sigma_{LN}(K,T)=\sigma$, for any $T>0$.

The equivalent log-normal variance can be expanded in a short maturity expansion as
\begin{equation}\label{SigExpansion}
\Sigma^2_{LN}(K,T) = \Sigma_0^2(e^x) + 2 \Sigma_0(e^x) \Sigma_1(K,T) T + 
O(T^2) \,.
\end{equation}
When including the higher order corrections in $T$, it will be seen to be convenient to work with the log-moneyness parameter $x = \log(K/A_{\rm fwd})$. This is expanded as
$x = \log k + O(rT)$, such that at leading order in $T$ the log-moneyness and 
log-strike are equivalent. When working at higher orders in $T$ it is important to keep track of the higher order terms in this relation.

In this paper, we compute the $O(T)$ term in the expansion (\ref{SigExpansion}). 
This correction can be expanded in powers of log-moneyness as
\begin{equation}
\frac{1}{\sigma^2} 
2 \Sigma_0(e^x) \Sigma_1(K,T) T =  (\sigma^2 T) (s_0 + s_1 x + O(x^2)) + 
 (r T) (r_0 + r_1 x + O(x^2)) \,.
\end{equation}
We give explicit results for the first two terms in this expansion $r_{0,1}, s_{0,1}$.
The higher order terms in the $x$-expansion can be evaluated using the same approach.

We summarize the expansion of the Asian implied volatility including the $O(T)$ term in the following result. This is the main result of this paper. 
%\subsection{New result}

\begin{proposition}\label{prop:main}
Assume that the asset price follows the Black-Scholes model 
$dS_t = \sigma S_t dW_t + r S_t dt$.
The equivalent log-normal variance of an 
Asian option with strike $K$ and maturity $T$ is 
\begin{eqnarray}\label{exact1}
\Sigma_{\rm LN}^2(K,T) &=& \sigma^2 \Big\{ 
\underbrace{\frac{x^2}{2J_{\rm BS}(e^x)} }
\underbrace{ - \frac{61}{9,\!450} (\sigma^2 T) + \frac{1}{12} (rT)}  \\
& & \hspace{1cm} O(1) \hspace{1.5cm} O(T) \nonumber \\
& &  \hspace{0.5cm} + 
\underbrace{\Big[ - \frac{34}{23,\!625} (\sigma^2 T) \Big] x }
+ 
O(T x^2) + O(T^2)\Big\} \,. \nonumber\\
& & \hspace{2.5cm} O(T x) \nonumber
\end{eqnarray}
where $x=\log (K/A_{\rm fwd})$ is the option log-moneyness.
\end{proposition}

In particular, for an at-the-money Asian option with strike $K=A_{\rm fwd}$ we have
\begin{equation}\label{SigLNATM}
\Sigma^2_{\rm LN}(K=A_{\rm fwd},T) = \sigma^2 
\left( \frac13 -  \frac{61}{9,\!450} (\sigma^2  T)
+ \frac{1}{12} (rT) + O(T^2) \right)\,.
\end{equation}

%Furthermore, the subset of higher order corrections, proportional to $(rT)^n$, is known to all orders in $T$. 

\subsection{Standardization}
\label{sec:1.1}

It is convenient to standardize the pricing problem by reducing it to the study of the distributional properties of the quantity
\begin{equation}
A_t^{(\mu)} = \int_0^t e^{2(B_s+\mu s)} ds \,,
\end{equation}
where $B_t$ is a standard Brownian motion.
Using the rescaling property of the Brownian motion $B_{\lambda t} = \sqrt{\lambda} B_t$, 
the time-average of the geometric Brownian motion in (\ref{ATdef}) can be expressed in terms of $A_t^{(\mu)}$ as \cite{Geman1993}
\begin{equation}\label{ATstandard}
A_T = \frac{4S_0}{\sigma^2 T} A_{\frac14 \sigma^2 T}^{(\frac{2r}{\sigma^2}-1)}
= S_0 \frac{1}{\tau} A_\tau^{(\mu)}\,, \quad \tau = \frac14 \sigma^2 T\,, \quad \mu = \frac{2r}{\sigma^2}-1 \,.
\end{equation}

The prices of fixed strike Asian options with averaging over the period $[0,T]$
and strike $K$ can be expressed in terms of the standardized average
$\frac{1}{\tau} A_\tau^{(\mu)}$ as \cite{Geman1993}
\begin{eqnarray}
&& C(K,T) = e^{-rT} \mathbb{E}\left[\left( \frac{1}{T} A_T - K \right)^+\right] = 
S_0 e^{-rT} c(k,\tau) \\
&& P(K,T) = e^{-rT} \mathbb{E}\left[\left( K - \frac{1}{T} A_T \right)^+\right] =
S_0 e^{-rT} p(k,\tau) \,.
\end{eqnarray}
with $k=K/S_0$ and
\begin{equation}
c(k,\tau) := \mathbb{E}\Big[\Big( \frac{1}{\tau} A_\tau^{(\mu)} - k\Big)^+\Big] \,,\quad
p(k,\tau) := \mathbb{E}\Big[\Big( K - \frac{1}{\tau} A_\tau^{(\mu)} \Big)^+\Big] \,.
\end{equation}

% moved here equations (32), (33)
The normalized Asian options $c(k,\tau),p(k,\tau)$ correspond to volatility 
$\sigma=2$ and drift $r=\mu+1$, with $S_0=1$.
The equivalent log-normal volatility for the normalized options 
$\sigma_{LN}(k,\tau)$ is defined as
\begin{equation}
c(k,\tau) = C_{BS}(k,\tau; a_{\rm fwd}^{(\mu)}, \sigma_{\rm LN}(k,\tau))\,,\quad 
a_{\rm fwd}^{(\mu)} := \mathbb{E}\Big[ \frac{1}{\tau} A_\tau^{(\mu)}\Big] 
= \frac{e^{2(\mu+1)\tau}-1}{2(\mu+1)\tau} \,.
\end{equation}
This is rescaled to the general $S_0,r,\sigma$ BS model as
\begin{equation}\label{rescaling}
\Sigma_{LN}^2(K,T;S_0,\sigma) = \frac14 \sigma^2 \sigma_{LN}^2\Big(\frac{K}{S_0},\frac14\sigma^2 T\Big)\,.
\end{equation}

\subsection{Outline}

The paper is structured as follows. The starting point is the 
asymptotic expansion of the density of the time average of the gBM
$\frac{1}{t} A_t^{(\mu)}$ given in Proposition 6 in \cite{HWexp},
obtained from a small time expansion of the Hartman-Watson distribution.
We collect the relevant properties of this expansion in Section~\ref{sec:2}.

The proof of the main result, Proposition~\ref{prop:main}, is given in 
Section~\ref{sec:3} and is divided into three parts, organized as separate sections.
In Section~\ref{sec:3.1} we compute the subleading correction to the price of Asian options in the BS model in a small maturity expansion. This is used in
Sections~\ref{sec:3.2} and \ref{sec:3.3} to obtain
the $O(T)$ correction to the equivalent log-normal implied volatility of an Asian option by an application of the Gao-Lee transfer result \cite{GaoLee}. 
The result can be used as the basis of a simple numerical pricing approximation for Asian options in the Black-Scholes model. 
In Section~\ref{sec:4} we present numerical tests on benchmark cases in the literature, showing that adding the subleading correction improves the agreement of the asymptotic expansion with the benchmark evaluations.

%%%%%%%%%%%%%%%%%%%%%%%%%%%%%%%%%%%%
\section{The asymptotic distribution of $\frac{1}{t} A_t^{(\mu)}$}
\label{sec:2}

The starting point for our analysis is a result obtained in \cite{HWexp} for the leading 
asymptotics of the density of $\frac{1}{t} A_t^{(\mu)}$ as $t\to 0$. This was obtained by
applying Laplace asymptotic methods to a one-dimensional integral
giving this density (due to Yor \cite{Yor1992}), combined with an asymptotic
expansion for the  Hartman-Watson distribution $\theta(r,t)$ as $t\to 0$.

%%% added a summary of \theta, F, G
For completeness, we summarize the main results of \cite{HWexp}. 
Denote the density of the normalized average of the geometric Brownian motion (gBM)
\begin{equation}\label{fdef}
\mathbb{P}\Big( \frac{1}{\tau} A_\tau^{(\mu)} \in da \Big) = f(a,\tau) \frac{da}{a}\,.
\end{equation}
The density is expressed as \cite{Yor1992}
\begin{equation}\label{fsol}
f(a,\tau) = e^{-\frac12 \mu^2 \tau} a^{\mu-1}
\int_0^\infty \rho^\mu e^{-\frac{1+a^2 \rho^2}{2a \tau}}
\theta(\rho/\tau, \tau) \frac{d\rho}{\rho}
\end{equation}

The Hartman-Watson function is defined by the integral
\begin{equation}
\theta(r,t) = \frac{r}{\sqrt{2\pi^3 t}} e^{\frac{\pi^2}{2t}}
\int_0^\infty e^{-\frac{\xi^2}{2t}} e^{-r \cosh \xi} \sinh \xi \sin\frac{\pi\xi}{t} d\xi
\end{equation}

Proposition 1 of \cite{HWexp} gives an expansion for this function
as $t\to 0$ at fixed $\rho = r t$ 
\begin{equation}
\theta(\rho/t, t) = \frac{1}{2\pi t} e^{-\frac{1}{t}[F(\rho) - \frac{\pi^2}{2}]} G(\rho) 
( 1 + \vartheta(\rho,t))
\end{equation}
where the functions $F(\rho), G(\rho)$ are known in closed form, and the error
term is bounded as $|\vartheta(\rho,t) | \leq \frac{1}{70} t$.

The density of the time integral of the gBM $f(a,\tau)$ 
can be approximated with the (properly normalized) leading term of this expansion as 
$f(a,\tau) = f_0(a,\tau) ( 1 + \varepsilon(a,\tau) )$ with
\begin{equation}\label{f0def}
f_0(a,\tau) := \frac{1}{n(\tau)} f_{\rm HW}(a,\tau) 
\end{equation}
where
\begin{equation}\label{fHWdef}
f_{HW}(a,\tau) := \frac{1}{2\pi \tau} e^{-\frac12 \mu^2 \tau} a^{\mu-1}
\int_0^\infty \rho^\mu G(\rho) e^{-\frac{1}{\tau} H(\rho,a)} \frac{d\rho}{\rho} \,.
\end{equation}
We denoted here
\begin{equation}\label{Hdef}
H(\rho) = \frac{1+a^2 \rho^2}{2a} -\frac{\pi^2}{2} + F(\rho) 
\end{equation}
and
\begin{equation}
n(\tau) = \int_0^\infty f_{HW}(a,\tau) \frac{da}{a}
\end{equation}
is a normalization factor which ensures that $f_0(a,\tau)$ is normalized as
$\int_0^\infty f_0(a,\tau) \frac{da}{a}=1$.

%Give here error bounds on $f_0(a,\tau)$.} 
The error of the approximation (\ref{f0def}) is bounded by the following result.

\begin{proposition}
The error of the approximation (\ref{f0def}) is bounded as
\begin{equation}\label{fbounds}
f_0(a,\tau) \frac{- \frac{1}{35}\tau}{1 + \frac{1}{70}\tau} \leq f(a,\tau) - f_0(a,\tau) \leq 
f_0(a,\tau) \frac{\frac{1}{35}\tau }{1 - \frac{1}{70}\tau } \,.
\end{equation}
\end{proposition}

\begin{proof}
Using (\ref{fsol}) we have
\begin{eqnarray}\label{dfHWerr}
&& |f(a,\tau ) - f_{HW}(a,\tau )| \leq e^{-\frac12\mu^2\tau} a^{\mu-1} 
\int_0^\infty \rho^\mu e^{-\frac{1+a^2\rho^2}{2a\tau}}
|\theta(\rho/\tau,\tau) - \frac{1}{2\pi\tau} e^{-\frac{1}{\tau}[F(\rho) - \frac{\pi^2}{2}]} G(\rho)| \frac{d\rho}{\rho} \\
&& \hspace{2cm} \leq
 \frac{1}{2\pi \tau} 
e^{-\frac12 \mu^2 \tau} a^{\mu-1}
\int_0^\infty \rho^\mu G(\rho) e^{-\frac{1}{\tau} H(\rho,a)} |\vartheta(\rho,\tau) |
\frac{d\rho}{\rho}
\leq
\frac{1}{70}\tau f_{HW}(a,\tau) \nonumber
\end{eqnarray}
In the last step we used the error bound $|\vartheta(\rho,\tau)| \leq \frac{1}{70}\tau$.

In a similar way we have
\begin{eqnarray}
&& | 1- n(\tau)  | = \left|\int_0^\infty (f(a,\tau) - f_{HW}(a,\tau ) ) \frac{da}{a} \right|
\leq  \int_0^\infty | f(a,\tau) - f_{HW}(a,\tau ) | \frac{da}{a} \\
&& \qquad \leq \frac{1}{70}\tau 
\int_0^\infty f_{HW}(a,\tau )  \frac{da}{a} = \frac{1}{70} \tau n(\tau ) \,,\nonumber
\end{eqnarray}
where we used (\ref{dfHWerr}) in the last step.

From these two inequalities we get

\begin{equation}
\frac{1}{1+\frac{1}{70} \tau} f(a,\tau) \leq f_{HW}(a,\tau) \leq 
\frac{1}{1-\frac{1}{70} \tau} f(a,\tau)
\end{equation}
and
\begin{equation}
\frac{1}{1+\frac{1}{70} \tau} \leq n(\tau) \leq \frac{1}{1-\frac{1}{70} \tau}
\end{equation}
Taking their ratio gives
\begin{equation}
\frac{1 - \frac{1}{70} \tau}{1+\frac{1}{70} \tau} f(a,\tau) \leq f_{0}(a,\tau) \leq 
\frac{1+\frac{1}{70} \tau}{1-\frac{1}{70} \tau} f(a,\tau) \,.
\end{equation}
These inequalities can be inverted to give bounds for $f(a,\tau)$ in terms of 
$f_0(a,\tau)$, which can be expressed as the error bounds (\ref{fbounds}).

\end{proof}

%%%%%%%%%%%%%%%%%%%%%%%%%%%%%%%%%%%%%%%%%
%%%%%%%%%%%%%%%%%%%%%%%%%
In this paper we are interested in the small-$\tau$ expansion of the integral (\ref{fHWdef}).
The application of Laplace asymptotic methods to this integral 
gives the more explicit result. 

\begin{proposition}\label{prop:Aexp}[Proposition 6 in \cite{HWexp}]
We have the $\tau\to 0$ asymptotics
\begin{eqnarray}\label{fHW0}
&& f_{HW}(a,\tau) =
\frac{1}{\sqrt{2\pi \tau}}
g(a,\mu)  e^{- \frac{1}{\tau} J(a)}  ( 1 + O(\tau) )
\end{eqnarray}
where 
\begin{equation}\label{gadef}
g(a,\mu) := 
(a\rho_*)^\mu G(\rho_*) \frac{1}{\sqrt{H''(\rho_*)}} \frac{1}{\rho_*}\,.
\end{equation}

We denote here $J(a) \equiv \inf_{\rho\geq 0} H(\rho) = H(\rho_*)$ 
and $\rho_* = \mbox{argmin} H(\rho)$.
From (\ref{Hdef}) it follows that the minimizer $\rho_*$
depends only on $a$ but not on $\mu$.

\end{proposition}

%\subsection{Series expansion for $g(a,\mu)$}

The leading asymptotics of the function $f_{HW}(a,\tau)$ in (\ref{fHW0})
depends on two functions $J(a)$ and
$g(a,\mu)$. The properties of the function $J(a)$ were studied in Sec.~4.1 of
\cite{HWexp} where it was shown that it is simply related to the rate function
$J_{BS}(k)$ appearing in the short-maturity asymptotics of Asian options, as
\begin{equation}
J(a) = \frac14 J_{BS}(a) \,.
\end{equation}

The following expansion of $g(a,\mu)$ was obtained in Proposition 10 of \cite{HWexp}. 
The coefficient $c_3$ quoted below is new. 

\begin{proposition}\label{prop:8}
The function $g(a,\mu)$ has the expansion
\begin{align}\label{gamuexp}
g(a,\mu) &= 
e^{\mu \log a + (\mu-1)\log\rho_*(a)} G(\rho_*) \frac{1}{\sqrt{H''(\rho_*)}}\\
&= \frac{\sqrt3}{2} e^{c_1 \log a + 
c_2 \log^2 a + c_3 \log^3 a + O(\log^4 a)}
\nonumber \,.
\end{align}
The first few coefficients $c_i$ are
\begin{eqnarray}\label{c1}
&& c_1 = \frac34 (\mu+1) - \frac45 \\
\label{c2}
&& c_2 = -\frac{3}{80}(\mu+1) + \frac{57}{1,400}\\
\label{c3}
&& c_3 = \frac{1}{350} (\mu+1) - \frac{1}{875} \,.
\end{eqnarray}

\end{proposition}

%%%%%%%%%%%%%%%%%%%%%%%%%%%%%%%%%%%%%%%
\section{Subleading corrections to the Asian implied volatility}
\label{sec:3}

The equivalent log-normal volatility $\Sigma_{\rm LN}(K,T)$ of an Asian option in the Black-Scholes model can be expanded in powers of maturity as
\begin{equation}
\Sigma_{\rm LN}(K,T) = \Sigma_0(K/S_0) + T \Sigma_1(K,S_0) + O(T^2) \,.
\end{equation}

The leading term in this expansion is determined from the short-maturity 
asymptotics of the Asian option prices \cite{SIFIN}
\begin{equation}
\Sigma_0^2(k) = \sigma^2 \frac{\log^2 k}{2J_{\rm BS}(k)} =
\frac13 \sigma^2 \left( 1 + \frac15 \log k - \frac{1}{84} \log^2 k - \frac{17}{10,500} \log^3 k + O(\log^4 k) \right) \,.
\end{equation}
See Proposition 18 in \cite{SIFIN}, where the equivalent log-normal volatility is denoted
$\Sigma_{\rm LN}(K,S_0)$.

We compute here the subleading term of $O(T)$ to the equivalent log-normal volatility. 
The proof proceeds in three steps. In the first step (Sec.~\ref{sec:3.1}) we compute the short maturity asymptotics for the reduced Asian option prices. 
In the second step (Sec.~\ref{sec:3.2}) we determine the equivalent log-normal volatility in the driftless case $r=0$, and in the third step (Sec.~\ref{sec:3.3}) 
a non-zero interest rate is added. 

\subsection{Short maturity asymptotics for Asian option prices}
\label{sec:3.1}

In this section we 
use the asymptotic distribution of the time average $\frac{1}{t} A_t^{(\mu)}$
in Proposition~\ref{prop:Aexp} to compute the price of OTM Asian options in the 
Black-Scholes model.  

\begin{proposition}\label{prop:10}
The leading asymptotics for the OTM Asian options with reduced
strike $k = K/S_0$ and maturity $\tau$ is
\begin{eqnarray}\label{cexp}
&& c(k,\tau) =  
\sqrt{\frac{\tau^3}{2\pi}}
\frac{g(k,\mu)}{k [J'(k)]^2} e^{-\frac{1}{\tau} J(k)} (1 + O(\tau)) \,,\quad k \geq 1  \\
\label{pexp}
&& p(k,\tau) = \sqrt{\frac{\tau^3}{2\pi}}
\frac{g(k,\mu)}{k [J'(k)]^2} e^{-\frac{1}{\tau} J(k)}(1 + O(\tau)) \,,\quad k \leq 1  \,.
\end{eqnarray}

\end{proposition}

\begin{proof}

The reduced Asian option price is  expressed as an integral over the exact distribution of the time-integral of the gBM $f(a,\tau)$, defined in (\ref{fdef})
\begin{equation}
c(k,\tau) = \int_0^\infty (a - k)^+ f(a,\tau) \frac{da}{a}\,.
\end{equation}

We derive an approximation for $c(k,\tau)$ by performing two successive approximations: 

i) replace $f(a,\tau)$ with its $\tau\to 0$
leading order approximation $f_0(a,\tau)$ defined in (\ref{f0def}).
Define the corresponding approximation for the option prices
\begin{equation}\label{c0def}
c_0(k,\tau) := \int_0^\infty (a - k)^+ f_0(a,\tau) \frac{da}{a}\,,\quad
p_0(k,\tau) := \int_0^\infty (k - a )^+ f_0(a,\tau) \frac{da}{a} \,.
\end{equation}

The error of this approximation is bounded using the error bound (\ref{fbounds})
as
\begin{equation}\label{cbounds}
- \frac{1}{35}\tau \frac{1}{1+\frac{1}{70} \tau} c_0(k,\tau) \leq c(k,\tau) - c_0(k,\tau) \leq 
\frac{1}{35}\tau \frac{1}{1 - \frac{1}{70} \tau} c_0(k,\tau) 
\end{equation}
and analogous for $p(k,\tau)$.
The approximation error is bounded in absolute value as
\begin{equation}
|c(k,\tau) - c_0(k,\tau) | \leq \frac{1}{35}\tau \frac{1}{1 + \frac{1}{70} \tau} c_0(k,\tau)\,.
\end{equation}

The approximation $c_0(k,\tau)$ is expressed as a double integral with an integrand
known in closed form. This can be easily evaluated numerically, and offers a simple approximation for pricing Asian options in the Black-Scholes model, with controlled approximation error. Tests of this approach in Section 5 of \cite{Nandori2021} demonstrate good agreement with the precise benchmarks of \cite{Linetsky2004}.

ii) Next we compute the leading approximation for $c_0(k,\tau)$ as $\tau \to 0$
using standard Laplace asymptotic methods for integrals. 
The result we use is due to Erd\'elyi, see Sec.~2.4 in \cite{Erdelyi}, 
and appears as Theorem 8.1 in Olver \cite{Olver}.
We give a few details of the application of this result to the 
Asian call price asymptotics,
using the notations of Theorem~1.2.1 of Nemes \cite{Nemes},
which is reproduced in the Appendix. 
The theorem applies to our case with the substitutions: 
\begin{equation}
\lambda \mapsto 1/\tau\,,\quad
f(x) \mapsto J(a) \,,\quad 
g(x) \mapsto \frac{1}{\sqrt{2\pi \tau}} (a-k) g(a,\mu) \frac{1}{a}
\end{equation}
and $\alpha=1,\beta=2$. 
%The integration boundaries are $a\mapsto k, b\mapsto \infty$. 

The technical conditions of the theorem are satisfied:
i) The function $J(a)$ is increasing on the integration
interval $[k,\infty)$ with $k>1$.
ii) $J(a), g(a,\mu)$ are continuous functions on $a\in [k,\infty)$.
iii) The functions $J(a), g(a,\mu)$ can be expanded around $k>1$ as in (\ref{fgexp}). 
For $J_{BS}(x)$ this follows from the analyticity of this function for $x>0$, 
see Sec.~4.1 in \cite{Nandori2021}. A similar result holds for $g(a,\mu)$ and follows from the analyticity of $F(\rho),G(\rho)$ for real positive $\rho$ proved in Sec.~4.2 of \cite{Nandori2021}. The leading coefficients in the expansion (\ref{fgexp}) are
\begin{equation}
a_0 = J'(k)\,,\quad
b_0 = \frac{1}{\sqrt{2\pi \tau}} g(k,\mu) \frac{1}{k} \,.
\end{equation}

iv) The integrals (\ref{c0def}) converge. See Sec.~5 of \cite{Nandori2021} for numerical evaluations of $c_0(k,\tau)$.

At leading order the Laplace asymptotic expansion (\ref{Laplaceexp}) gives 
$c_0(k,\tau) = c_0^L(k,\tau) (1 + O(\tau))$ with
\begin{equation}
c_0^L(k,\tau) := e^{-\frac{1}{\tau} J(k) } \frac{d_0}{(1/\tau)^2}  = 
\sqrt{\frac{\tau^3}{2\pi}}
\frac{g(k,\mu)}{k [J'(k)]^2} e^{-\frac{1}{\tau} J(k)} 
\end{equation}
where we used %$d_0$ is given by (\ref{d0})
\begin{equation}\label{d0sol}
d_0 = \frac{b_0}{a_0^2} =\frac{1}{\sqrt{2\pi \tau}} g(k,\mu) \frac{1}{k(J'(k))^2}\,.
\end{equation}
The correction to the leading order term is of order 
$|c_0(k,\tau) - c_{0}^L(k,\tau)| =
c_{0}^L(k,\tau) (1 + O(\tau))$.

The combined error of the two approximations is 
\begin{eqnarray}
&& | c(k,\tau)  - c_{0}^L(k,\tau) |\leq |c(k,\tau) - c_0(k,\tau) | + 
|c_0(k,\tau) - c_{0}^L(k,\tau) | \\
&& \qquad \leq c_0(k, \tau) (1 + O(\tau) ) + c_{0}^L(k,\tau) (1 + O(\tau))
= c_{0}^L(k,\tau) (1 + O(\tau)) \,.\nonumber
\end{eqnarray}

This reproduces the quoted result (\ref{cexp}). 
The Asian put option result (\ref{pexp}) is obtained in a similar way. 
\end{proof}

%%%%%%%%%%%%%%%%%%%%%%%%%%%%%%%%%%%%%%%%
\subsection{The driftless case $r=0$}
\label{sec:3.2}

We start with the simpler case $r=0$. 
Recall that in terms of the normalized parameters introduced in Sec.~\ref{sec:1.1} 
this corresponds to $\mu=-1$. 
In the next step (Sec.~\ref{sec:3.3}) we include the contribution of a non-zero 
interest rate.

The starting point is the asymptotic result for option 
prices of Proposition~\ref{prop:10}. The leading asymptotic result for an
OTM Asian call option has the form $c(k,\tau) = \tau^{3/2} h(k) 
e^{-\frac{1}{\tau} J(k)}$ with $h(k) = \frac{1}{\sqrt{2\pi}} \frac{g(k,\mu)}{k[J'(k)]^2}$.
Recall $k=K/S_0$.

The small-$\tau$ asymptotics of the Asian option can be expressed as an expansion for the log-price $L=- \log c(k,\tau)$. 
Using the notations of Gao and Lee \cite{GaoLee}, the first terms of this asymptotics are
\begin{equation}\label{Lexp}
L = - \log c(k,\tau) = \frac{1}{\tau} J(k) - \frac32 \log \tau + \alpha_0(k)
\end{equation}
with $\alpha_0(k) := - \log h(k)$ which is expanded as
\begin{eqnarray}\label{alpha0}
&& \alpha_0(k) =
- \log\left( \frac{16\sqrt3}{18 \sqrt{2\pi}} \cdot 
\frac{k}{\log^2 k} \right) + 2\log \left(1 - \frac{3}{10} \log k + O(\log^2 k)
\right) - \sum_{i=1}^\infty c_i \log^i k 
\end{eqnarray}
The second term in this expression is the contribution from $J'(k)$ 
which is expanded by substituting
the series expansion (\ref{JTaylor}) for $J_{BS}(k)$ and differentiating term by term
\begin{equation}
J'(k) = \frac14 J'_{\rm BS}(k) = \frac{3}{4k} \log k \left( 1 - \frac{3}{10} \log k + 
\frac{109}{1,050} \log^2 k + O(\log^3 k) \right) \,.
\end{equation}

The third term in (\ref{alpha0}) contains the contribution of the exponential factor
for $g(k,\mu)$ in Proposition \ref{prop:8}, which is determined by the
coefficients $c_i$ appearing in the expansion of the exponent
around $k=1$.

By Corollary 7.4 in Gao, Lee \cite{GaoLee}, the asymptotic implied variance is
\begin{equation}\label{Sig7.4}
\sigma^2_{\rm LN}(k,\tau) = \frac{\log^2 k}{2J(k)} - \frac{\log^2 k}{4 J^2(k)}
\left( \log k + \log \frac{\log^2 k}{16\pi} + 2 \alpha_0(k) - 3\log J(k)
\right) \tau + O(\tau^2)
\end{equation}

%This expression has apparent singular behavior as $k\to 1$ because of the
%$\log k$ terms in the brackets. Recalling that the rate function $J(k)$ has 
%the expansion 
%\begin{equation}
%J(k) = \frac{3}{8} \log^2 k \left( 1 - \frac{1}{5} \log k + 
%\frac{109}{2100} \log^2 k + O(\log^3 k) \right)
%\end{equation}
%it is clear that a $\log k$ term in the expression in brackets in the second term would give a singular contribution $\sim 1/\log k$. We show next that all such singular terms as $k\to 1$ vanish and the result is finite.

The expression in the brackets in the second term of (\ref{Sig7.4}) 
is expanded around $k=1$ as
\begin{eqnarray}\label{Bexp}
&& B(k) := \log k + \log \frac{\log^2 k}{16\pi} + 2 \alpha_0(k) - 3\log J(k)\\
&& \quad\,\, = 
b_0 + b_1 \log k + b_{1L} \log \log^2 k + b_2 \log^2 k + O(\log^3 k)\,.\nonumber
\end{eqnarray}

Expanding in $\log k$ gives the coefficients
\begin{eqnarray}
&& b_0=0 \\
\label{b1sol}
&& b_1 = - (1 + \frac35 + 2 c_1 ) = - \frac32 (\mu+1)\\
&& b_{1L}=0 \\
\label{b2sol}
&& b_2 = \frac{293}{2,\!100} - 2c_2 = \frac{61}{1,\!050} + \frac{3}{40} (\mu+1)\,.
\end{eqnarray}

We keep the terms proportional to $\mu+1$, although they vanish for
$r=0$, in order to keep track of their
contributions for general $r$ in Sec.~\ref{sec:3.3}.

In the third line we used $c_1 = - \frac45 + \frac34 (\mu+1)$ from (\ref{c1}) 
to obtain the null result for $b_{1L}$.
In the last line we used $c_2 = \frac{57}{1,\!400} - \frac{3}{80}(\mu+1)$ from (\ref{c2}).

Substituting the expansion (\ref{Bexp}) into (\ref{Sig7.4}) we get the expansion 
in $\log k$
\begin{eqnarray}\label{ExpSig2}
\sigma^2_{\rm LN}(k,\tau) &=& \frac{\log^2 k}{2J(k)} + \left\{
 - \frac{16 b_1}{9 \log k} - \frac{16}{45} \left(2 b_1 + 5 b_2\right) \right. \\
& & \left.- \frac{8(17 b_1 + 420 b_2 + 1,\!050 b_3)}{4,\! 725} \log k +
O(\log^2 k) \right\} \tau + O(\tau^2) \,.\nonumber
\end{eqnarray}

Note the presence of a singular term $1/\log k$ in the subleading volatility proportional to $b_1$; since $b_1 = -\frac32(\mu+1)$ this divergent term vanishes for the 
driftless gBM case $\mu=-1$.
After a more careful analysis in Sec.~\ref{sec:3.3} it will be seen to cancel also for the gBM with non-zero drift.

%We conclude that $\sigma_{\rm LN}(k,\tau)$ is well-behaved and finite 
%as $k\to 1$ in the driftless case $\mu=-1$. 

Substituting the expressions for $b_i$ from (\ref{b1sol}), (\ref{b2sol}) 
into (\ref{Sig7.4})
gives an explicit result for the ATM implied variance to $O(\tau)$ for the driftless 
case $\mu=-1$
\begin{eqnarray}\label{GenSig}
&& \sigma^2_{\rm LN}(k,\tau)|_{\mu=-1} = \frac{x^2}{2J(e^x)} + 
\left(- \frac{488}{4,\!725}  + O(x ) 
\right)\tau + O(\tau^2) \,. 
\end{eqnarray}

Rescaling to general $(\sigma,r,T)$ using (\ref{rescaling}) gives the first terms in
the equivalent log-normal volatility of an Asian option stated in 
Proposition~\ref{prop:main}. 
Keeping only the ATM expression for the $O(T)$ term this is 
\begin{eqnarray}\label{ATMSig}
\Sigma_{\rm LN}^2\Big(\frac{K}{S_0},T\Big)|_{r=0} =
\Sigma_0^2\Big(\frac{K}{A_{\rm fwd}}\Big) + 
\sigma^2 \left( - \frac{61}{9,\!450} (\sigma^2 T) + O(x) \right) + O(T^2) \,.
\end{eqnarray}
%%%%%%%%%%%%%%%%%%%%%%%%%%%%%%%%%%%%

\subsection{Including a non-zero interest rate $r$}
\label{sec:3.3}

In the last step of the proof we include the contribution of the rate $r$
into the log-strike definition. We will show that this ensures the cancellation of the divergent term proportional to $\mu+1$ in (\ref{ExpSig2}), and adds a new finite term proportional to this factor.

Expanding the log-moneyness of the Asian option to order $O(\tau)$ we have
\begin{equation}
x = \log\frac{k}{a_{\rm fwd}^{(\mu)}} = \log k - (\mu+1) \tau + O(\tau^2)
\end{equation}
where the forward price is the average of $A_\tau^{(\mu)}$ 
obtained using
$r=\mu+1$ in standardized units
\begin{equation}\label{afwd}
a_{\rm fwd}^{(\mu)} := \mathbb{E}\left[\frac{1}{\tau} A_\tau^{(\mu)}\right] 
= \frac{1}{2(1+\mu)\tau} \left( e^{2(1+\mu) \tau} - 1\right) 
\simeq 1 + (\mu+1)\tau + O(\tau^2) \,.
\end{equation}

\begin{proof}[Proof of the general result $r\neq 0$.]
The small-$\tau$ expansion of $L = - \log c(k,\tau)$ at fixed $x$ is obtained by replacing 
$k \to (1 + (\mu+1)\tau)e^x$ into (\ref{Lexp}). To $O(\tau)$ it is sufficient
to replace $\log k \mapsto x + (\mu+1) \tau$. Expanding in $\tau$ we find
\begin{align}
L &= \frac{1}{\tau} J(e^x) + (\mu+1) e^x J'(e^x) - \frac32\log\tau + \alpha_0(e^x) + O(\tau) \\
&= \frac{1}{\tau} J(e^x) + (\mu+1) \left( \frac34 x - \frac{9}{40} x^2 +
\frac{109}{1,\!400} x^3 + O(x^4) \right)  - \frac32 \log \tau + \alpha_0(e^x) + O(\tau)\nonumber\\
&:= \frac{1}{\tau} J(e^x) - \frac32 \log \tau + \tilde\alpha_0(e^x) + O(\tau)\,. \nonumber
\end{align}
In the last step we absorbed the second term in the second line into 
$\tilde \alpha_0(e^x)$.
The effect of the new term is to shift the coefficients $b_k$ defined in (\ref{Bexp}) as
\begin{eqnarray}\label{b1tilde}
&& b_1 \to \tilde b_1 := b_1 + \frac32 (\mu+1) = 0 \\
\label{b2tilde}
&& b_2 \to \tilde b_2 := b_2 - \frac{9}{20}(\mu+1) = \frac{61}{1,050} - \frac38(\mu+1)\,.
\end{eqnarray}

Substituting $b_i\to \tilde b_i$ into (\ref{ExpSig2}) gives
\begin{eqnarray}\label{ExpSig2x}
\sigma^2_{LN}(e^x,\tau) &=& \frac{x^2}{2J(e^x)} + \left\{
 - \frac{16 \tilde b_1}{9 x} - \frac{16}{45} \left( 2 \tilde b_1 + 5 \tilde b_2\right) \right. \\
& & \left.- \frac{8(17 \tilde b_1 + 420 \tilde b_2 + 1,\!050 \tilde b_3)}{4,\!725} x + O(x^2) \right\} 
\tau + O(\tau^2) \nonumber
\end{eqnarray}

Substituting here the explicit results for $\tilde b_i$ from (\ref{b1tilde}), (\ref{b2tilde}) 
we get
\begin{equation}
\sigma^2_{\rm LN}(e^x,\tau) = \frac{x^2}{2J(e^x)} 
+ \left(-\frac{488}{4,\!725} + \frac23(\mu+1) - \frac{544}{23,\!625} x
+ O(x^2) \right) \tau + O(\tau^2) \,.
\end{equation}

All singular terms as $x \to 0$ cancel out, and the 
Asian volatility $\sigma_{\rm LN}(e^x,\tau)$ is finite and well-defined at the
ATM point $x=0$.

Rescaling to arbitrary volatility $\sigma$ and the actual maturity $T$ using
$\tau \to \frac14 \sigma^2 T,\mu+1 \to \frac{2r}{\sigma^2}$,
see (\ref{rescaling}), gives the final result
\begin{eqnarray}\label{final}
\Sigma^2_{\rm LN}(K,T,S_0) &=& \sigma^2 \Big\{ \frac{x^2}{2J_{\rm BS}(e^x)} 
- \frac{61}{9,\!450} (\sigma^2 T) + \frac{1}{12} (rT)  \\
& &  \hspace{0.5cm} -
 \frac{34}{23,\!625} (\sigma^2 T) x 
+ 
O(x^2 T) + O(T^2)\Big\} \,. \nonumber
\end{eqnarray}

This concludes the proof of Proposition~\ref{prop:main}.
\end{proof}

Numerically the ATM implied variance is 
$\frac13 - 0.00645 (\sigma^2 T) + 0.083 (rT)$, such that 
the $O(rT)$ term dominates the subleading contribution for most 
realistic values of the model parameters. 
The correction linear in log-moneyness is $- 0.00144 (\sigma^2 T) x$.

\subsection{Consistency check and an improved estimate}
The small maturity limit of the equivalent log-normal volatility of an Asian option 
in the Black-Scholes model has been obtained in \cite{PZAsian} in a modified
small maturity limit $\sigma^2 T \to 0$ taken at fixed $\rho = rT$. 
This limit takes into account interest rates effects; more precisely it includes 
corrections of the order $O((rT)^n)$ to all orders in $n$.

The result is given in Proposition 19 of \cite{PZAsian} and we denote it as
$\Sigma_{\rm LN,\rho}(K,\rho)$. As $\rho\to 0$, this reduces to $\Sigma_0(K/S_0)$ 
given in (\ref{Sigma0def}).

The Asian volatility $\Sigma_{\rm LN,\rho}(K,\rho)$ includes corrections of the order $O((rT)^n)$ to all orders in $n$.
At the ATM point this function simplifies and is given by (see
equation~(125) in \cite{PZAsian})
\begin{equation}\label{SigATM}
\Sigma_{LN,\rho}(K=A_{\rm fwd}, \rho) = \sigma \frac{S_0}{A_{\rm fwd}} \sqrt{v(\rho)} 
= \sigma \frac{\rho}{e^\rho-1} \sqrt{v(\rho)} 
\end{equation}
with $A_{\rm fwd} = S_0 \frac{e^\rho-1}{\rho}$ and 
\begin{equation}
v(\rho) := \frac{1}{\rho^3} \Big( \rho e^{2\rho} - \frac32 e^{2\rho} + 2 e^\rho - \frac12\Big)=
\frac13 + \frac{5}{12} \rho + \frac{17}{60} \rho^3 + O(\rho^4)
\end{equation}
We will use this result to test the coefficient of the $O(rT)$ term in (\ref{final}).

Squaring (\ref{SigATM}) and expanding in $\rho$ gives
\begin{equation}
\Sigma_{LN,\rho}^2(K=A_{\rm fwd},\rho) = 
\sigma^2 \frac{\rho^2}{(e^\rho-1)^2} v(\rho) =
\sigma^2 \left( \frac13 +\frac{1}{12}\rho + \frac{1}{180} \rho^2 + O(\rho^3) \right)\,.
\end{equation}
This reproduces the $+\frac{1}{12}(rT)$ correction in Eq.~(\ref{final}). 

This suggests an improved approximation for $\Sigma_{LN}(K,T)$, obtained by replacing $\sigma^2 \frac{x^2}{2J_{BS}(e^x)}$ in (\ref{final}) with 
$\Sigma^2_{LN,\rho}(K,\rho)$.
This approximation includes interest rates effects through the leading order term, by taking into account corrections of order $O((rT)^n)$ to all orders, in addition to the 
$O(\sigma^2 T)$ correction computed here. This approximation is somewhat heuristic, as it neglects e.g. corrections of $O((\sigma^2 T)^n)$ with $n>1$. 
It should be useful in situations when the $O(rT)$ corrections dominate numerically
over $O(\sigma^2 T)$.

We denote this improved next-to-leading order (NLO) estimate as
$\Sigma_{LN,NLO}(K,T)$. It is given explicitly by
\begin{equation}\label{NLO}
\Sigma^2 _{\rm LN,NLO}(K,T) := \Sigma_{LN,\rho}^2(K,\rho) +
\sigma^2 \Big( - \frac{61}{9,\!450} (\sigma^2 T) - \frac{34}{23,\!625} (\sigma^2 T) 
\log \frac{K}{A_{\rm fwd}} \Big) \,.
\end{equation}
This approximation can be further improved by adding terms of higher order in
log-moneyness $x=\log(K/A_{\rm fwd})$. Keeping terms up to the linear term 
in $x$ should give an accurate approximation in a region of strikes sufficiently close to the ATM point. 

\section{Numerical examples}
\label{sec:4}

In this section we present a few numerical tests of our results.
We can price Asian options by substituting the equivalent
log-normal volatility $\Sigma_{\rm LN}(K,T)$ of Proposition~\ref{prop:main} into the Black-Scholes formula. This gives the Asian prices
\begin{equation}
C(K,T) = e^{-rT} [ A_{\rm fwd} \Phi(d_1) - K \Phi(d_2) ]\,,\quad
P(K,T) = e^{-rT} [K \Phi(-d_2) - A_{\rm fwd} \Phi(-d_1) ]\,,
\end{equation}
with $A_{\rm fwd}$ given in (\ref{Afwd}) and
\begin{equation}
d_{1,2} = \frac{1}{\Sigma_{\rm LN}(K,T)\sqrt{T}} \left( \log\frac{A_{\rm fwd}}{K}
\pm \frac12 \Sigma_{\rm LN}^2(K,T) T \right)\,.
\end{equation}

Using this approach we evaluate the seven benchmark cases given in Linetsky 
\cite{Linetsky2004}, and compare them against the precise results obtained in this 
paper using a spectral expansion approach. Table~\ref{tab:1} shows the results for Asian option prices obtained from the leading order asymptotic result of \cite{PZAsian} (column $C_0(K,T)$) and the improved results obtained keeping also the $O(T)$ subleading correction. The columns $C_1^{ATM}(K,T)$ and $C_1^{\rm lin}(K,T)$
show the results obtained by keeping only the ATM subleading correction $O(T)$ in (\ref{exact1}), and including also the term linear in log-strike $O(T x)$, respectively.
The last column shows the benchmark results from \cite{Linetsky2004} obtained
using a precise spectral expansion. 
In brackets we show the relative error with respect to the benchmark, for each approximation.

The results are shown also in Figure~\ref{Fig:cases}. The plots show the equivalent log-normal volatility $\Sigma_{LN}(K,T)$ including only the ATM subleading correction (first line of (\ref{exact1})) vs $k=K/S_0$ (solid curves), comparing with the leading order result $\Sigma_0(k)$ (dashed curves). 
The result for $\Sigma_{\rm LN}(K,T)$ (\ref{exact1}) depends only on $(\sigma,r,T)$, 
so cases 4,5,6 have a common curve. The dots show the precise benchmark values in the last column of Table~\ref{tab:1} converted to volatility. 
The vertical lines
show the ATM strike $A_{\rm fwd}(S_0)/S_0$.

The agreement with the benchmark results improves significantly when including the subleading correction, especially in cases with large $rT$. 
The error of $C_1^{ATM}(K,T)$ is below 0.02\% in all cases, and becomes even smaller for
$C_1^{\rm lin}(K,T)$ when including the subleading skew contribution.

The improved approximation (\ref{NLO}) which includes corrections of order $O((rT)^n)$ to all orders is tested in Table~\ref{tab:2} against the same benchmark cases from
\cite{Linetsky2004}. This approximation is expected to perform better for cases with large $rT$. This is confirmed indeed, as seen for case 3 which has the largest values of this parameter $rT=0.18$ and agrees with the benchmark to five digits. The approximation error is below 0.02\% in all cases. 

\begin{table}
\caption{\label{tab:1} 
Seven benchmark cases for Asian options. $C_0(K,T)$ denotes the
Asian options obtained using the leading approximation $\Sigma_0(k)$ for the
equivalent log-normal Asian volatility.
$C_1^{\rm lin}(K,T)$ denotes the Asian option price obtained by including 
both terms in the subleading correction (\ref{exact1}), and 
$C_1^{ATM}(K,T)$ includes only the ATM subleading correction.
The last column shows the precise results of \cite{Linetsky2004}
obtained by a spectral expansion. 
Relative errors relative to the benchmarks are shown in brackets [bps].}
\begin{center}
\begin{tabular}{|c|cccc|cccc|}
\hline
Case & $k$ & $r$ & $\sigma$ & $T$ & $C_0(K,T)$ 
         & $C_1^{ATM}(K,T)$ & $C_1^{lin}(K,T)$ & benchmark  \\
\hline\hline
1 & 1 & 0.02 & 0.10 & 1 & 0.055923 (-11.2) & 0.055986 (0.0) & 0.055986 (0.0)
   & 0.055986 \\
2 & 1 & 0.18 & 0.30 & 1 & 0.217054 (-61.0) & 0.218362 (-1.1) & 0.218364 (-1.0)
   & 0.218387 \\
3 & 1& 0.0125 & 0.25 & 2 & 0.172163 (-6.2) & 0.172268 (-0.1) & 0.172269 (0.0)
   & 0.172269 \\
4 & $\frac{2}{1.9}$ & 0.05 & 0.50 & 1 & 0.192895 (-14.4) & 0.193176 (0.1) 
   & 0.193173 (0.0) & 0.193174 \\
5 & 1 & 0.05 & 0.50 & 1 & 0.246125 (-11.8) & 0.246412 (-0.2) & 0.246415 (0.0)
   & 0.246416 \\
6 & $\frac{2}{2.1}$ & 0.05 & 0.50 & 1 & 0.305927 (-9.6) & 0.306211 (-0.3) 
   & 0.306220 (0.0)   & 0.306220 \\
7 & 1 & 0.05 & 0.50 & 2 & 0.349314 (-22.3) & 0.350077 (-0.5) & 0.350093 (0.0)
   & 0.350095 \\
\hline
\end{tabular}
\end{center}
\end{table}

\begin{table}
\caption{\label{tab:2} 
The predictions for Asian options prices obtained using the improved approximation
for the equivalent log-normal volatility $\Sigma_{LN,NLO}(K,T)$ in (\ref{NLO}) including terms of all orders in $O((rT)^n)$ (NLO), comparing with the benchmarks of Linetsky \cite{Linetsky2004}. The scenarios are the same as
in Table~\ref{tab:1}. Last row shows the relative error in basis points.}
\begin{center}
\begin{tabular}{|c|ccc|cccc|}
\hline
Case & 1 & 2 & 3 & 4 & 5 & 6 & 7 \\
\hline\hline
NLO & 0.055986 & 0.218385 & 0.172268 & 0.193188 & 0.246409 & 0.306193 
        & 0.350060 \\
Linetsky & 0.055986 & 0.218387 & 0.172269 & 0.193174 & 0.246416 & 0.306220
        & 0.350095 \\
err [bp] & 0 & 0.09 & 0.05 & 0.64 & -0.32 & -1.24 & -1.60 \\
\hline
\end{tabular}
\end{center}
\end{table}

\begin{figure}
    \centering
      \includegraphics[width=2in]{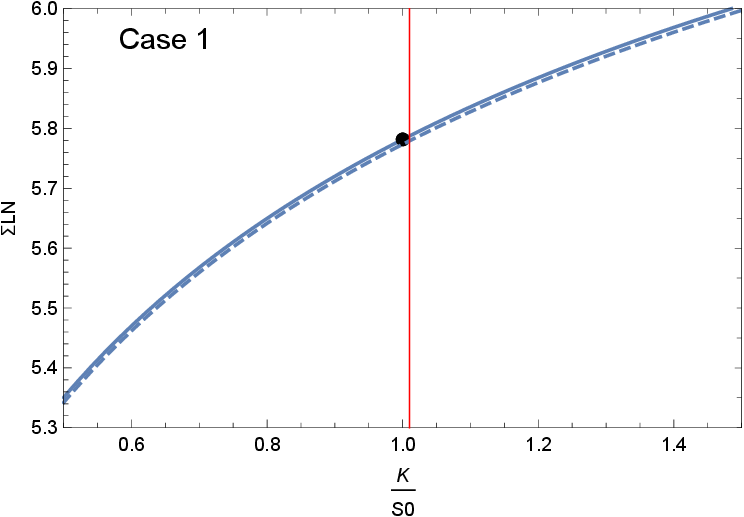}
   \includegraphics[width=2in]{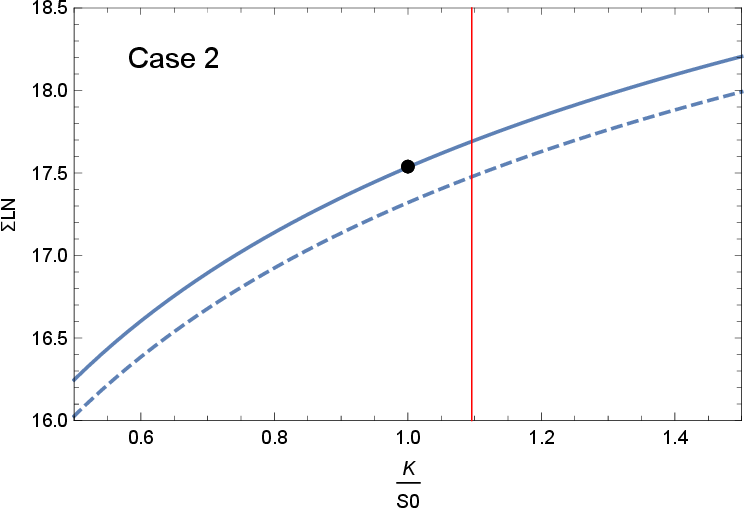}
   \includegraphics[width=2in]{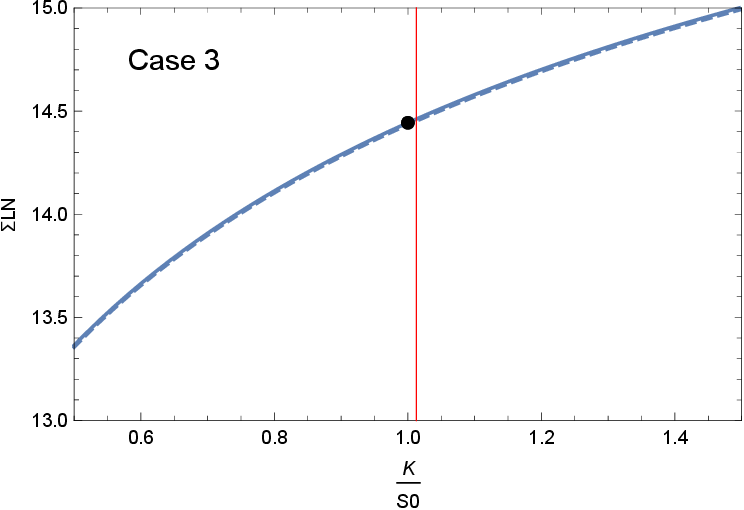}
   \includegraphics[width=2in]{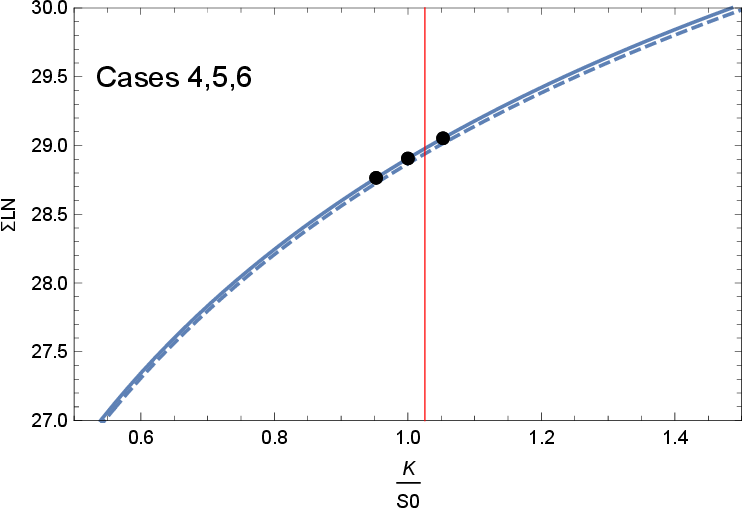}
   \includegraphics[width=2in]{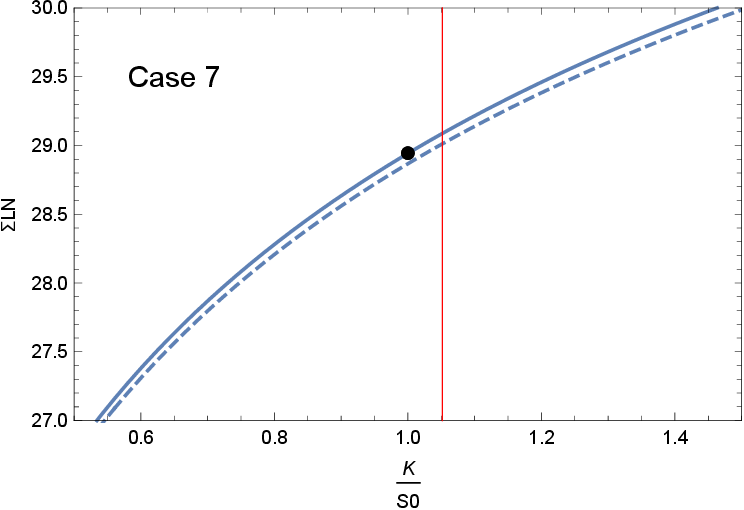}
    \caption{ Asian volatilities $\Sigma_{\rm LN}(K,T)[\%]$  vs $k=K/S_0$ 
for the seven scenarios in Table~\ref{tab:1}. The dashed curves show the leading order Asian volatility $\Sigma_0(K/S_0)$ and the solid curves include the ATM subleading correction in (\ref{exact1}). The dots show the benchmark cases in the last column of Table~\ref{tab:1}. 
The vertical line shows the ATM strike $A_{\rm fwd}(S_0)/S_0$.}
\label{Fig:cases}
 \end{figure}

\newpage
%%%%%%%%%%%%%%%%%%%%%%%%%%%%%%%%%%%%%%%%%%
\textbf{Note added (July 2024).}
We summarize here the result for the Asian implied variance, in a form which makes explicit the dependence on the coefficients $c_i$ and is easier to extend to higher orders. 
This is used to extend Proposition \ref{prop:main} by including also the convexity term of the subleading Asian implied variance.

The reduced Asian implied variance to order $O(\tau)$ is
\begin{eqnarray}
&& \sigma_{LN}^2(k,\tau) = \frac{\log^2 k}{2J(k)} \\
&& + \Big\{ \frac{4}{4,\!725} (1,\!051 + 1,\!680 c_1 + 4,\!200 c_2 - 315(\mu+1) ) \nonumber \\
&& \quad + \frac{8}{23,\!625} (-91 + 170 c_1 + 4,\!200 c_2 + 10,\!500 c_3) \log k \nonumber \\
&&  + \frac{1}{18,\!191,\!250} ( -250,\!193 - 517,\!440 c_1 + 1,\!047,\!200 c_2 + 25,\!872,\!000 c_3 \nonumber \\
&& \hspace{2.5cm} + 64,\!680,\!000 c_4 - 39,\!270(\mu+1)) \log^2 k \nonumber \\
&& \quad + O(\log^3 k) \Big\} \tau + O(\tau^2)\nonumber
\end{eqnarray}
The first term is the leading implied variance, to all orders in $\log k$. The second, third and fourth terms give the ATM level, skew and convexity of the subleading implied variance,
respectively.

Substituting here the coefficients $c_{1-3}$ given in (\ref{c1}) - (\ref{c3}) and
%\footnote{The coefficient $c_4$ is new.}
\begin{eqnarray}
c_4 = \frac{11}{22,\!400} (\mu+1) - \frac{2,\!897}{3,\!080,\!000}
\end{eqnarray}
gives the following result for the subleading Asian implied variance.

\begin{eqnarray}
&& \sigma_{LN}^2(e^x,\tau) = \frac{x^2}{2J(e^x)} 
+ \Big(  -\frac{488}{4,\!725} + \frac23(\mu+1) - \frac{544}{23,\!625} x \\
&& \hspace{2cm} + \, \Big( \frac{1,\!657}{259,\!875} - \frac{5}{252} (\mu+1) \Big) x^2 + O(x^3) \Big) \tau + O(\tau^2) \nonumber 
\end{eqnarray}
where $x=\log (K/A_{\rm fwd})$ is the option log-moneyness.

Rescaling to arbitrary Black-Scholes parameters gives the analog of equation (12), including also the convexity term.

\begin{eqnarray}
\Sigma_{\rm LN}^2(K,T) &=& \sigma^2 \Big\{ 
\underbrace{\frac{x^2}{2J_{\rm BS}(e^x)} }
\underbrace{ - \frac{61}{9,\!450} (\sigma^2 T) + \frac{1}{12} (rT)} \nonumber \\
\label{exact2}
& & \hspace{1cm} O(1) \hspace{1.5cm} O(T) \\
& &  \hspace{0.5cm} + 
\underbrace{\Big[ - \frac{34}{23,\!625} (\sigma^2 T) \Big] x } \nonumber \\
& & \hspace{2.5cm} O(T x) \nonumber \\
& & \hspace{0.5cm} +
\underbrace{\Big[ \frac{1,\!657}{4,\!158,\!000} (\sigma^2 T) - \frac{5}{2,\!016} (rT) \Big] x^2 }
+ 
O(T x^3) + O(T^2)\Big\} \,. \nonumber\\
& & \hspace{2.5cm} O(T x^2) \nonumber
\end{eqnarray}

\newpage
\appendix

\section{Asymptotic expansion for integrals}

The following theorem is due to Erd\'elyi and is given in Sec.~2.4 of 
\cite{Erdelyi} (p. 36). It appears as Theorem 8.1 in Chapter~3.8 of Olver \cite{Olver}.
For convenience we quote it below in the notations of Theorem 1.2.1 in 
Nemes \cite{Nemes}.

\begin{theorem}
Consider the integral
\begin{equation}
I(\lambda) = \int_a^b e^{-\lambda f(x)} g(x) dx
\end{equation}

Assume that:

(i) $f(x) > f(a)$ for all $x\in (a,b)$.

(ii) $f'(x), g(x) $ are continuous in a neighborhood of $a$.

(iii) the following expansions hold
\begin{eqnarray}\label{fgexp}
&& f(x) = f(a) +\sum_{k=0}^\infty a_k(x-a)^{\alpha+k} \\
&& g(x) = g(a) +\sum_{k=0}^\infty b_k(x-a)^{\beta+k-1} \,.\nonumber
\end{eqnarray}

(iv) $I(\lambda)$ converges absolutely for all sufficiently large $\lambda$.

Then
\begin{equation}\label{Laplaceexp}
I(\lambda) = e^{-\lambda f(a)} \sum_{n=0}^\infty 
\Gamma\Big(\frac{n+\beta}{\alpha}\Big)
\frac{d_n}{\lambda^{(n+\beta)/\alpha}} \,.
\end{equation}
The coefficient of the leading order term is
\begin{equation}\label{d0}
d_0 = \frac{b_0}{\alpha a_0^{\beta/\alpha}} \,.
\end{equation}

\end{theorem}

%%%%%%%%%%%%%%%%%%%%%%%%%%%%%%%%%%%%%%%%%%%%%%%%%%%%%%%%%%%%%%%%%%%%%%

%%%%%%%%%%%%%%%%%%%%%%%%%%%%%%%%%%%%%%%%%%%%%%%%%%%%%%%%%%%%%%%%%%%%%%%%%

\end{document}